\documentclass[11pt,letterpaper]{article}

\usepackage[margin=1in]{geometry}

\usepackage{graphicx} 
\usepackage{authblk}
\usepackage{amsthm}
\usepackage{amsmath}
\usepackage{amssymb,mathtools}
\usepackage{graphicx}
\usepackage{thmtools} 
\usepackage{thm-restate}
\usepackage{enumerate}
\usepackage{bm}

\usepackage{graphicx} 
\usepackage{mathtools}

\usepackage{hyperref}
\usepackage{tikz-cd}

\newtheorem{theorem}{Theorem}[section]
\newtheorem{claim}[theorem]{Claim}
\newtheorem{proposition}[theorem]{Proposition}
\newtheorem{lemma}[theorem]{Lemma}
\newtheorem{corollary}[theorem]{Corollary}

\theoremstyle{definition}
\newtheorem{Definition}[theorem]{Definition}
\newtheorem{example}[theorem]{Example}

\newtheorem{remark}[theorem]{Remark}

\renewcommand{\log}{\ln}

\usepackage{tikz}
\foreach \x in {A,...,Z}{%
    \expandafter\xdef\csname m\x\endcsname{\noexpand\mathbf{\x}}
}
\foreach \x in {a,...,t,v,w,x,y,z}{%
    \expandafter\xdef\csname m\x\endcsname{\noexpand\mathbf{\x}}
}
\newcommand{\0}{\mathbf{0}}

\newcommand{\pr}[2]{\left\langle #1, #2 \right\rangle}
\newcommand{\R}{\mathbb{R}}

\newcommand{\Zp}{\mathbb{Z}_{\ge0}}
\newcommand{\Rp}{\mathbb{R}_{\ge0}}

\newcommand{\defeq}{\stackrel{\scriptscriptstyle{\mathrm{def}}}{=}}

\newcommand{\agents}{{\mathcal{A}}}
\newcommand{\items}{\mathcal{E}}
\newcommand{\NSW}{\operatorname{NSW}}

\newcommand{\fP}{{\mathcal{P}}}

\newcommand{\supply}{s}

\newcommand{\demand}[1]{{\mathcal{D}_{#1}}}

\newcommand{\hide}[1]{}

\newcommand{\ee}{\mathrm{e}}

\usepackage[style=alphabetic,giveninits, natbib=true, maxbibnames=99]{biblatex}
\DeclareFieldFormat*{title}{#1}
\usepackage{hyperref}
\usepackage{tikz-cd}
\usepackage{booktabs}

\usepackage[linesnumbered,ruled,vlined]{algorithm2e}
\newcommand\blfootnote[1]{%
  \begingroup
  \renewcommand\thefootnote{}\footnote{#1}%
  \addtocounter{footnote}{-1}%
  \endgroup
}

\title{Tight Efficiency Bounds for the Probabilistic Serial \\ and Related Mechanisms\blfootnote{J.~Garg was supported by NSF Grant CCF-2334461. Y. Tao was supported by the National Natural Science Foundation of China (Grant No. 62472029) and National Key R\&D Program of China(2023YFA1009500). 
}}
\author[1]{Jugal Garg}
\author[2]{Yixin Tao} 
\author[3]{L{\'{a}}szl{\'{o}} A. V{\'{e}}gh}
\affil[1]{University of Illinois at Urbana-Champaign, USA}
\affil[2]{ITCS, Key Laboratory of Interdisciplinary Research of Computation and Economics\\ Shanghai University of Finance and Economics, China }
\affil[3]{Hertz Chair for Algorithms and Optimization, University of Bonn, Germany}
\date{}

\begin{document}

\maketitle
\begin{abstract}
The Probabilistic Serial (PS) mechanism---also known as the simultaneous eating algorithm---is a canonical solution for the random assignment problem under ordinal preferences. It guarantees envy-freeness and ordinal efficiency in the resulting random assignment. However, under cardinal preferences, its efficiency may degrade significantly: it is known that PS may yield allocations that are $\Omega(\ln{n})$-worse than Pareto optimal, but whether this bound is tight remained an open question. 

Our first result resolves this question by proving that the PS mechanism guarantees $(\ln n+1)$-approximate Pareto efficiency under cardinal preferences. The key part of our analysis shows that PS achieves a logarithmic $(\ln n + 1)$-approximation to the maximum Nash welfare, in stark contrast to the $O(\sqrt{n})$ loss that can arise in utilitarian social welfare.
Our results also extend to the more general submodular setting introduced by Fujishige, Sano, and Zhan (ACM TEAC 2018). 
In addition, we present a polynomial-time algorithm that computes an allocation which is envy-free and $\ee^{1/\ee}$-approximately Pareto-efficient, answering an open question posed by Tr\"obst and Vazirani (EC 2024).

The PS mechanism also applies to the allocation of chores instead of goods. We prove that it guarantees an $n$-approximately Pareto-efficient allocation in this setting, and that this bound is asymptotically tight. This result provides the first known approximation guarantee for computing a fair and efficient allocation in the random assignment problem with chores under cardinal preferences.
\end{abstract}

\section{Introduction}
We consider the problem of allocating indivisible items to agents in a fair and efficient manner in the random assignment problem. Let $\agents$ denote a set of $n$ agents, and $\items$ a set of $m\ge n$ items.\footnote{The case $m<n$ can be reduced to $m=n$ by introducing dummy items that have zero utility for all agents.} We focus on allocations that may assign at most one item to each agent, and monetary transfers are not allowed. Thus, every feasible allocation corresponds to a matching between agents and items. A standard way to ensure fairness is to introduce randomization and assign items via lotteries. Such random assignment problems have been extensively investigated in the literature (e.g., \cite{aziz2019random,Bogomolnaia2001,hylland1979efficient, roth1993stable}). Two fundamental models are commonly used to express agent preferences: the \emph{cardinal} and \emph{ordinal} models. In the cardinal model, each agent $i\in\agents$ specifies a non-negative utility $u_{ij}\ge 0$ for each item $j\in\items$. In contrast, the ordinal model assumes that agents only provide strict preference orderings $(\succ_i)_{i\in\agents}$ over the items, without quantifying the intensity of their preferences. We say that the cardinal preferences $(u_{ij})_{j\in\items}$ are \emph{consistent} with the ordinal preference ordering $\succ_i$, if  $j\succ_i j'$ implies $u_{ij}\ge u_{ij'}$. Given such a cardinal utility vector, an \emph{induced preference order} refers to any strict ordering consistent with it.

We consider randomized allocations of items to agents, represented as $\mx\in\Rp^{\agents\times \items}$, where $x_{ij}$ denotes the probability of agent $i$ receiving item $j$; we let $\mx_i=(x_{ij})_{j\in\items}$. Thus, we require $\sum_{j\in\items} x_{ij}\le 1$ for all $i\in\agents$, and $\sum_{i\in\agents} x_{ij}\le 1$ for all $j\in\items$. By the Birkhoff--von Neumann theorem, one can always find a lottery, i.e., probability distribution over (partial) assignments such that agent $i$ receives item $j$ with probability $x_{ij}$. The expected cardinal utility of agent $i$ is thus $u_i(\mx_i)\coloneqq \sum_{j\in\agents} u_{ij}x_{ij}$. 

A mechanism is said to be \emph{envy-free} in the cardinal model if $u_i(\mx_i)\ge u_i(\mx_{i'})$ for any $i,i'\in \agents$. In the ordinal preference model, envy-freeness can be defined using the concept of stochastic dominance. A probability distribution $p$ on $\items$ \emph{stochastically dominates} another distribution $q$ with respect to the strict order $\succ$, if for any item $j\in\items$, $\sum_{k\succeq j} p_k\ge \sum_{k\succeq j} q_k$, with the inequality being strict for at least one item $j$. A mechanism is \emph{envy-free} in the ordinal model if there are no two agents $i,i'\in\agents$ such that the allocation $\mx_{i'}$  stochastically dominates $\mx_i$ with respect to agent $i$'s preference order $\succ_i$. This in particular implies that $i$ does not envy $i'$ for any cardinal preferences $u_{ij}$ that are consistent with $\succ_i$.

In a seminal paper, \citet{Bogomolnaia2001} introduced the \emph{probabilistic serial} (PS) mechanism for randomized assignment in the ordinal preference model. This is based on a natural \emph{simultaneous eating algorithm}: each agent starts accruing a fraction of their most preferred item at the same rate. When an item is fully consumed (i.e., its unit supply is exhausted), agents move to their next most preferred available item. This process continues until either each agent has consumed a total of one unit or all items are exhausted.

\citet{Bogomolnaia2001} showed that the eating algorithm finds an envy-free randomized assignment. They also showed that it is \emph{ordinally efficient}, that is, no other randomized assignment stochastically dominates it for all agents. 

The PS mechanism is among the most extensively studied mechanisms in the literature and has inspired a large body of follow-up work.   
\citet{bogomolnaia2012probabilistic} showed that PS is the unique mechanism satisfying ordinal efficiency, envy-freeness, and a property known as bounded invariance. Related and stronger characterizations were later provided by Hashimoto, Hirata, Kesten, Kurino and \"Unver~[\citeyear{hashimoto2014two}]. The PS mechanism has also been extended to a variety of more general settings; e.g.,~\cite{ashlagi2020assignment,athanassoglou2011house,aziz2014generalization,bogomolnaia2015random,Budish2013,katta2006solution, yilmaz2009random,Fujishige2018,Aziz2022,SchulmanV15}; see these works for additional references on generalizations of the PS mechanism and related developments.

\subsection{Efficiency bounds in the cardinal utility model} The simultaneous eating algorithm (i.e., the PS mechanism) can also be applied in the cardinal utility model by running it with an ordering consistent with the agents' utility values $(u_{ij})_{j\in\items}$. In this formulation, the output depends only on the ordinal rankings of items, not the magnitude of utilities. As a result, the algorithm returns the same allocation for two instances where all agents share the same preference ordering, even if the ratios $u_{ij}/u_{ij'}$ differ significantly---for instance, being exponentially large in one instance and only slightly above $1$ in another. It is therefore not surprising that the strong guarantees of the ordinal setting do not carry over to the cardinal setting. While the resulting allocation remains envy-free, it could be Pareto dominated by another allocation. This raises a natural and fundamental question of how much efficiency is lost by the PS mechanism under cardinal preferences, and whether this loss can be arbitrarily large.

To quantify this loss, Immorlica, Lucier, Weyl, and Mollner [\citeyear{Immorlica2017}] introduced the notion of \emph{$\gamma$-approximately Pareto efficient lotteries}: For $\gamma\ge 1$, a randomized assignment $\mx\in\R^{\agents\times\items}$ is said to be $\gamma$-approximately Pareto efficient (or $\gamma$-Pareto efficient for short), if there does not exist another assignment $\my$ such that $u_i(\my_i)\ge \gamma u_i(\mx_i)$ for all $i\in\agents$, with strict inequality for at least one agent. They studied the approximate efficiency of various mechanisms and showed, in particular, that the PS mechanism is not $c\log(n) $-Pareto efficient for some constant $c>0$; that is, there exist instances where all agents can achieve utility a factor $c\log(n)$ higher under some other allocation than under the PS allocation. However, it remained an open question whether this bound is worst possible.

An alternate notion of efficiency is social welfare, particularly,  utilitarian social welfare, defined as $\sum_{i\in\agents} u_i(x_i)=\sum_{i\in\agents\, ,\, j\in \items} u_{ij}x_{ij}$, which corresponds to the total utility of the fractional matching. In this context, \citet{zhang2022tight} showed that the eating algorithm approximates utilitarian welfare within a factor of $\Theta(\sqrt{n})$: that is, it achieves at least a $1/O(\sqrt{n})$ fraction of the maximum possible welfare, and this factor is asymptotically tight. 
An ordinal analogue of utilitarian welfare approximation is the notion of rank approximation, introduced by \citet{chakrabarty2014welfare}. They showed that the PS mechanism achieves a rank approximation guarantee of at least $\Omega(\sqrt{n})$, which is later shown to be also tight~\cite{zhang2022tight}.  

\medskip

In light of the above results on the inefficiency of the PS mechanism, our first theorem may appear surprising. We show that the inefficiency example of \cite{Immorlica2017} is, in fact, asymptotically the worst possible. Moreover, we strengthen this result in two aspects: We prove the same guarantee for \emph{Maximum Nash welfare}, i.e., the geometric mean of the utilities, see \eqref{eq:nsw}, and we show that this guarantee extends to the more general submodular randomized assignment setting.

\begin{theorem}\label{theorem:eating-guar}
The PS mechanism approximates the maximum Nash welfare within a factor of $\ln(n)+1$. Consequently, it is $(\ln(n)+1)$-approximately Pareto efficient.
\end{theorem}

\paragraph{Submodular randomized assignment problem}
In the standard random assignment model, any selection of $n$ items constitutes a feasible allocation. However, in many practical settings, additional constraints may apply. For example, availability may be restricted across different groups of items. For a structured family of such constraints, known as a hierarchy or laminar family, the PS mechanism was extended by Budish, Che, Kojima, and Milgrom~[\citeyear{Budish2013}]. A far-reaching further generalization was later given by Fujishige, Sano, and Zhan~[\citeyear{Fujishige2018}] that is particularly relevant to our work. In this model, the fractional assignment of items must be in a polymatroid (submodular polytope); the model is formally described in Section~\ref{sec:sub-prel}. The natural generalization of the mechanism in this context is the submodular eating algorithm, described in Algorithm~\ref{alg:submodular-eating}.

The Nash welfare maximization guarantee extends to this general submodular setting.
In fact, Theorem~\ref{theorem:eating-guar} stated earlier arises as a special case. This requires only the natural assumption that a full unit of each item can be consumed under the submodular constraints. 
\begin{theorem}\label{thm:submod-eating-upper} 
Consider the Submodular Randomized Assignment model with the assumption that a full unit of each item can be consumed under the submodular constraints.   
Then, the submodular eating algorithm approximates the maximum Nash welfare within a factor of $\ln(n)+1$. Consequently, it is $(\ln(n)+1)$-approximately Pareto efficient.
\end{theorem}

\subsection{Efficiency guarantees for envy-free mechanisms}
The PS mechanism guarantees envy-freeness in both the ordinal and cardinal models. However, while there may be a logarithmic loss in Pareto efficiency, it cannot exceed this bound, as shown in Theorem~\ref{thm:submod-eating-upper}, even in the more general submodular setting. 

On the other hand, the classical result by \citet{hylland1979efficient} shows the existence of envy-free and Pareto-efficient randomized allocations through a market-based mechanism based on \emph{Fisher markets} with linear utilities and matching constraints. This is done by showing the existence of a \emph{competitive equilibrium} in this market, which was recently shown to be PPAD-hard to compute by Chen, Chen, Peng, and Yannakakis~[\citeyear{chen2022computational}].  
Tr\"obst and Vazirani~[\citeyear{troebst2024cardinal}] further established that finding an envy-free and Pareto-efficient randomized assignment is also PPAD-hard. 

Given this intractability, we can relax the requirement and aim to find an allocation $\mx$ that is $\alpha$-envy-free, i.e., for all agents $i,i'\in\agents$, $u_i(\mx_i)\ge u_i(\mx_{i'})/\alpha$. Tr\"obst and Vazirani~[\citeyear{troebst2024cardinal}] provide a convex program whose optimal solutions are 2-approximately envy-free and Pareto-efficient.
Hence, while achieving both envy-freeness and Pareto-efficiency is PPAD-complete, one can efficiently find a Pareto-efficient allocation that satisfies envy-freeness up to a constant factor. 

This naturally raises the question, posed also in~\cite{troebst2024cardinal}: is the converse tradeoff achievable? That is, can we find in polynomial-time an envy-free allocation that is $\alpha$-approximate Pareto efficiency for some constant $\alpha$? This is particularly relevant in applications where fairness is prioritized over efficiency. Our next result answers this in the affirmative, even in the more general submodular setting. A more formal statement is given in Section~\ref{sec:envy-nash} as Theorem~\ref{thm:envy-main2}.

\begin{theorem}\label{thm:envy-main} 
There exists a polynomial-time algorithm that, for any $\varepsilon >0$, finds an envy-free and $(\ee^{1 / \ee}-\varepsilon)$-Pareto efficient allocation for the random assignment problem and for the more general submodular setting. 
 \end{theorem}

\paragraph{Envy-freeness and approximately Pareto efficiency for general utilities} One may ask the analogous question in the general setting of concave monotone nondecreasing utility functions.
Classical results in mathematical economics \cite{ArrowDebreu1954} imply the existence of a \emph{competitive equilibrium} in Fisher markets with such utilities. Competitive equilibrium allocations are Pareto-efficient, and if all budgets are equal, then they are also envy-free. We note that the model does not cover matching markets where the bundles $\mx_i$ of the agents must satisfy additional constraints $\sum_{j\in \items} x_{ij}\le 1$; the proof of existence for this case is due to \citet{hylland1979efficient}. Later, \citet{garg2022approximating} showed the existence of equilibria in a more general constrained market model. 

However, computing equilibria is PPAD-hard even for simple nonlinear concave utility functions \cite{CT2009,bei2019earning}. In contrast, the maximum Nash welfare allocation can be easily shown to be 2-approximately envy-free even in the general Fisher market; see~\cite{garg2025approximating}. 

Analogous to the random assignment problem, one may ask whether it is possible to efficiently find envy-free but approximately Pareto-efficient allocations in this setting for general utilities. One challenge here is that the approach used for the random assignment problem in Theorem~\ref{thm:envy-main} does not extend to this setting due to the non-convexity of the formulation (see Section~\ref{sec:Fisher} for details). 
In Section~\ref{sec:Fisher}, we address this problem using a different approach, obtaining a weaker approximation guarantee. 

\begin{theorem}\label{thm:Fisher}
There exists a polynomial-time algorithm that for any $\varepsilon>0$, finds a $(1+\varepsilon)$-envy-free and $(2+2\varepsilon)$-approximate maximum Nash welfare allocation for Fisher markets under monotone concave utility functions. Such an allocation is also $(2+2\varepsilon)$-approximately Pareto efficient.
\end{theorem}

\subsection{Efficiency of the PS mechanism for chores}\label{sec:intro-chores}
Fairly and efficiently allocating a set of chores among agents is a fundamental problem in fair division and has attracted substantial recent attention (see, e.g.,~\cite{bogomolnaia2017competitive,branzei2024algorithms,ChaudhuryGMM22,ChaudhuryKMN24,MS25,liu2024mixed} and references therein), though most work ignores matching constraints, which arise naturally in applications; see~\cite{GargTV25}. 

The simultaneous eating algorithm can also be applied in the chores setting, where agents incur \emph{disutilities} rather than utilities from receiving items. In this context, a set $\items$ of $m$ chores needs to be assigned to a set $\agents$ of $n$ agents in a fair and efficient manner such that each agent receives a total of one unit of chores.\footnote{It is natural to assume $m=n$.
Our bounds hold even without this assumption: if $m>n$, some chores remain unassigned; if $m<n$, the problem can be reduced to $m=n$ by introducing dummy chores with arbitrarily small disutility for all agents.\label{footnote2}} 

In the cardinal model, agent $i\in \agents$ has disutility $d_{ij} \ge 0$ over item $j\in\items$. Given a fractional allocation $\mx\in\Rp^{\agents\times \items}$, the expected disutility of agent $i$ is defined as $d_i(\mx_i):= \sum_{j\in\items} d_{ij} x_{ij}$. 
In the ordinal model, we assume that agents report only a strict preference order $\succ_i$ over the items, without quantifying the intensity of their preferences ($j\succ_i j'$ means $d_{ij}\le d_{ij'}$).  

A mechanism is said to be envy-free in the cardinal model if $d_i(\mx_i) \le d_i(\mx_{i'})$ for any $i, i'\in\agents$, i.e., no agent strictly prefers the bundle assigned to another agent in terms of lower expected disutility. In the ordinal model, it is envy-free if for any $i,i'\in\agents$, $\mx_{i'}$ does not stochastically dominate $\mx_{i}$ with respect to $\succ_i$. 
As in the goods setting, the PS mechanism yields an envy-free and ordinally-efficient allocation in the chores setting as well. However, its efficiency in the cardinal utility model remained unknown. 

Garg, Tr\"obst, and Vazirani~[\citeyear{GargTV25}] studied the chores setting under cardinal preferences and showed that the HZ mechanism, originally designed for goods, can be adapted to chores via a utility-shifting transformation. This adaptation preserves both envy-freeness and Pareto efficiency. However, as noted earlier, computing such an allocation is PPAD-hard~\cite{VaziraniY25,troebst2024cardinal}, and thus the mechanism is not computationally efficient. In the goods setting, maximizing Nash welfare yields a 2-approximately envy-free and Pareto efficient allocation \cite{garg2025approximating,troebst2024cardinal}. Unfortunately, this approach does not carry over to the chores setting~\cite{GargTV25}, which is known to be computationally and structurally more challenging, even in the absence of matching constraints. As a result, the problem of finding an allocation that is simultaneously $\alpha(n)$-approximately envy-free and $\alpha(n)$-approximately Pareto efficient for any function $\alpha(n)$ remained open. 

We bound the efficiency of the simultaneous eating algorithm in the cardinal utility model for the chores setting: we show that it guarantees $n$-approximate Pareto efficiency, and this bound is asymptotically tight. Notably, this is the first approximation guarantee for computing an approximately fair and efficient allocation achieved by any algorithm in this setting with any approximation guarantee. 

\begin{theorem}\label{thm:chores-main}
Assume all disutilities are positive: $d_{ij}>0$ for all agents $i$ and chores $j$. Then, 
the simultaneous eating algorithm is $n$-approximately Pareto efficient in the chores setting. Moreover, there exist instances in which the disutility of each agent under its allocation is at least $n/4$ times higher than under some other feasible allocation.  
\end{theorem}
The efficiency guarantee does not hold without the assumption on positive disutilities. As illustrated by Example~\ref{ex:zero-chore}, there can be instances with an optimal allocation giving 0 disutility to all agents while the eating algorithm gives some agents positive disutility. 

\subsection{Further related literature}
Another popular mechanism in the ordinal setting, beyond probabilistic serial (PS), is Random Priority (RP)~\cite{Moulin2018fair}. RP is strategyproof but does not satisfy envy-freeness, whereas PS guarantees envy-freeness but is not strategyproof. Despite these differences, the two mechanisms are closely related in certain settings. 

\citet{che2010asymptotic} showed that when there are multiple copies of each item, the outcomes of RP and PS converge as the number of copies tends to infinity. ~\citet{kojima2010incentives} further showed with sufficiently many copies of each item, truthfully reporting ordinal preferences becomes a weakly dominant strategy for an agent with a given expected utility function under the PS mechanism. More recently, Wang, Wei, and Zhang~[\citeyear{wang2020bounded}] analyzed the incentive ratio of the eating algorithm and showed that it is bounded above by $1.5$, quantifying the maximum benefit an agent can gain by misreporting their preferences.

In the cardinal setting, the classical solution concept of Competitive Equilibrium with Equal Incomes (CEEI) provides allocations that are both envy-free and Pareto optimal~\cite{VARIAN74}. The HZ mechanism builds on CEEI by incorporating matching constraints and has been proposed for a variety of applications, including course allocation and other assignment problems~\cite{Budish2011combinatorial,He2018pseudo}. Furthermore, CEEI and its extensions have been extensively studied in the context of both goods and chores, particularly in settings without matching constraints (see recent work~\cite{bogomolnaia2017competitive,ChaudhuryGMM22,ChaudhuryKMN24} and the references therein). 

Another prominent approach in the literature is the maximum Nash welfare (MNW) objective, which has received significant attention in fair division and social choice; see, e.g.,~\cite{Moulin03}. For indivisible goods with additive valuations, MNW allocations give the ``unreasonable'' fairness guarantee of envy-freeness up to one good (EF1)~\cite{CaragiannisKMPS19}. 

There is a strong connection between CEEI and maximum Nash welfare allocations: they coincide when agents have homogeneous concave utility functions. The paper \cite{garg2025approximating} analyzed the relationship of these outcomes, showing that any competitive equilibrium $\ee^{1/\ee}$-approximates MNW, and conversely, for the broad class of \emph{Gale substitute} utilities, the MNW allocation gives each agent at least half of their maximum possible utility in any equilibrium. 

Due to the computational intractability of computing HZ allocations, a line of work has explored Nash-bargaining-based solutions, which offer efficient computation while retaining desirable structural and fairness properties of equilibrium-based approaches (see e.g.,~\cite{GargTV23,HosseiniV22,panageas2024time}).

\section{Efficiency of the PS mechanism}
\subsection{Model and preliminaries}\label{sec:prelim}
We consider the random assignment problem, where $\agents$ denotes a set of $n$ agents and $\items$ a set of $m$ items. In the ordinal model, each agent provides a strict preference ordering $(\succ_i)_{i\in\agents}$ over the items, without quantifying the intensity of their preferences. In contrast, in the cardinal model, each agent $i\in\agents$ specifies a non-negative utility $u_{ij}\ge 0$ for each item $j\in\items$. The cardinal preferences $(u_{ij})_{j\in\items}$ are \emph{consistent} with the ordinal preference ordering $\succ_i$, if  $j\succ_i j'$ implies $u_{ij}\ge u_{ij'}$. 

A (randomized) allocation is represented by $\mx\in\Rp^{\agents\times \items}$, where $x_{ij}$ denotes the probability of agent $i$ receiving item $j$; we write $\mx_i=(x_{ij})_{j\in\items}$ for the allocation of agent $i$. We require that $\sum_{i\in\agents} x_{ij}\le 1$ for all $j\in\items$, and $\sum_{j\in\items} x_{ij}\le 1$ for all $i\in\agents$,
i.e., $\mx$ corresponds to a (partial) matching between agents and items. 

A mechanism is said to be \emph{envy-free} in the cardinal model if $u_i(\mx_i)\ge u_i(\mx_{i'})$ for any $i,i'\in \agents$. In the ordinal model, envy-freeness is defined via stochastic dominance: an allocation $\mx_{i'}$ stochastically dominates $\mx_i$ with respect to $\succ_i$, if, for every item $j\in \items$, $\sum_{k\succeq j} x_{i'k} \ge \sum_{k\succeq j} x_{ik}$, with strict inequality for some $j$. A mechanism is \emph{envy-free} in the ordinal model if no such domination occurs, which in particular implies envy-freeness for all cardinal utilities consistent with $\succ_i$.

The simultaneous eating algorithm (i.e., the PS mechanism) outputs an envy-free randomized allocation and is \emph{ordinally-efficient}, that is, no other randomized assignment stochastically dominates it for all agents. 

In the cardinal model, efficiency is captured by Pareto dominance: an allocation $\mx'$ Pareto dominates $\mx$ if $\sum_j u_{ij}x'_{ij} \ge \sum_j u_{ij}x_{ij}, \forall i\in\agents$, with strict inequality for at least one agent. A mechanism is Pareto-efficient in the cardinal model if no such domination occurs. The PS mechanism is not Pareto-efficient in the cardinal model, as shown in the simple example below.

\begin{example}\label{example1}
Consider a scenario with three agents and three items $\{a, b, c\}$. All agents share the same ordinal preferences, ranking the items $c \succ b \succ a$. The specific utilities for each agent-item pair are detailed in Table~\ref{tab:agent-preferences}.

Under the PS mechanism, each agent receives a $1/3$ share of every item. This yields a utility of $1.7$ for agents $1$ and $2$, and $2.3$ for agent $3$.

In contrast, consider an alternative envy-free allocation where agents $1$ and $2$ equally divide items $a$ and $c$, and agent $3$ is allocated the entirety of item $b$. Under this alternative, agents $1$ and $2$ each receive a utility of $2$, while agent $3$ receives $2.9$. Therefore, this alternative envy-free assignment Pareto-dominates the PS allocation by providing a strict utility improvement for all agents.
\begin{table}[t]
\centering
\begin{tabular}{@{}lccc@{}}
\toprule
\textbf{Agent} & \textbf{Item $a$} & \textbf{Item $b$} & \textbf{Item $c$} \\ \midrule
Agent $1$              & 1            & 1.1                      & 3                       \\
Agent $2$              & 1                       & 1.1            & 3                      \\
Agent $3$              & 1            & 2.9                       & 3                      \\ \bottomrule
\end{tabular}
\caption{Utility profiles for agents 1--3 across three items.}
\label{tab:agent-preferences}
\end{table}
\end{example}

To quantify this loss, \citet{Immorlica2017} introduced the notion of \emph{$\gamma$-approximate Pareto efficient} allocations: an allocation $\mx$ is $\gamma$-approximate Pareto efficient if no allocation $\my$ satisfies $u_i(\my_i)\ge \gamma u_i(\mx_i)$ for all $i\in\agents$, with strict inequality for at least one agent. 

The classical \emph{Nash welfare} function, defined for a randomized assignment $\mx\in\Rp^{\agents\times \items}$ as the geometric mean of agents' utilities:
\begin{equation}\label{eq:nsw}
   \NSW(\mx)= \prod_{i\in\agents} u_i(\mx_i)^{1/n}\, .
\end{equation}
The \emph{maximum Nash welfare} assignment is the random assignment $\my$ that maximizes $\NSW(\my)$. We say that an allocation $\mx$ \emph{$\gamma$-approximates} the maximum Nash welfare if $\NSW(\mx)\ge \NSW(\my)/\gamma$ for any random assignment $\my$. It is immediate that any $\gamma$-approximate maximum Nash welfare allocation is also $\gamma$-approximately Pareto efficient.

\subsubsection{Submodular randomized assignment problem}\label{sec:sub-prel}
A set function  $\rho\,:\,2^\items\to \Rp$ is \emph{submodular} if
$\rho(X)+\rho(Y)\ge \rho(X\cap Y)+\rho(X\cup Y)$ for all $ X,Y\subseteq \items$.
We also assume that $\rho$ is monotone, i.e., $\rho(X)\ge\rho(Y)$ if $X\supseteq Y$, and that $\rho(\emptyset)=0$.
The \emph{polymatroid associated with $\rho$} is defined as  
\begin{equation}\label{eq:submodu-polytope}
\fP(\rho)=\big\{\mz\in \R^\items\, :\, \mz\ge 0\, , \, \sum_{s\in S}z_s\le \rho(S)\quad\forall S\subseteq\items\big\}\, .
\end{equation}
Vectors $\mz\in\fP(\rho)$ are called \emph{independent}; a set $S\subseteq\items$ is independent if its incidence vector satisfies $\chi_S\in \fP(\rho)$.

In the \emph{Submodular Randomized Assignment Model}, a feasible assignment is a vector 
$\mx\in \Rp^{\agents\times \items}$ such that
\begin{equation}\label{eq:submod-allocation}
\sum_{j\in \items} x_{ij} \le 1\, ,\quad\forall i\in\agents\, ,\quad\mbox{and}\quad 
\sum_{i\in\agents} \mx_i\in \fP(\rho)\, .
\end{equation}
That is, we require that the aggregate consumption vector 
$\sum_{i\in\agents} \mx_i$ is independent. 
The standard assignment model corresponds to the function $\rho(S)=|S|$. A more general case, where each item $j$ is available in $\supply_j$ copies, can be modeled by $\rho(S)=\sum_{j\in S}\supply_j$. 

We call a mapping $\pi\,:\, \agents\to\items\cup\{\emptyset\}$ an \emph{independent assignment}, if every agent receives at most one item (or none), and the set of assigned items, with multiplicities, forms an independent set. 
The following proposition,  a generalization of the Birkhoff--von Neumann theorem asserts that a fractional allocation as in \eqref{eq:submod-allocation} can be decomposed into a lottery over independent assignments. It is a consequence of the integrality of the submodular intersection polyhedron, see e.g., \cite[Theorem 46.1]{schrijver}.

\begin{proposition}\label{prop:lottery}
Let $\rho\,:\,2^\items\to \Zp$ be an integer valued monotone submodular function, and let $\mx\in \Rp^{\agents\times \items}$ satisfy \eqref{eq:submod-allocation}. Then, there exists independent assignments $\pi^{(t)}$ and probabilities $p_t$,  $t=1,\ldots,k$, such that selecting the assignment  $\pi^{(t)}$ with probability $p_t$ will give item $j$ to agent $i$ with probability $x_{ij}$ for all $i\in\agents$, $j\in\items$.
\end{proposition}

\paragraph{The submodular eating algorithm}
 Consider the ordinal preference model where each agent has a strict preference order $\succ_i$ over $\items$. The simultaneous eating algorithm naturally extends to this setting as follows.

Let  $\mx\in \Rp^{\agents\times \items}$ denote the current (partial) allocation, initialized as $\mx=\0$. Throughout, we maintain that $\mx$ satisfies the feasibility constraint~\eqref{eq:submod-allocation}.
An item $j\in \items$ is called \emph{active with respect to $\mx$}, if there exists $\varepsilon>0$ such that $\sum_{i\in \agents}\mx_i+\varepsilon \cdot \chi_{\{j\}}\in \fP(\rho)$, where $\chi_{\{j\}}$ is the incidence vector of item $j$.  
To characterize the set of active items, let us introduce the concept of tight sets.
For $\mz\in\fP(\rho)$, we say that $S\subseteq \items$ is \emph{tight} w.r.t. $\mz$, if $\rho(S)=\sum_{s \in S} z_s$. The following fundamental property can be derived by the standard uncrossing argument.
\begin{proposition}\label{prop:uncrossing}
For any $\mz\in \fP(\rho)$, and tight sets $S,T\subseteq \items$, $S\cap T$ and $S\cup T$ are also tight. Consequently, there is a unique largest tight set $T_{\mz}$.
\end{proposition}
The following corollary is immediate.
\begin{corollary}\label{cor:tight-active}
Let
$\mx\in \Rp^{\agents\times \items}$ such that $\sum_{i\in\agents} \mx_i\in \fP(\rho)$.
An item $j$ is active w.r.t. $\mx$ if and only if $j\in\items\setminus T_{\mz}$ for $\mz=\sum_{i\in\agents} \mx_i$.
\end{corollary}

In each round, $\items\setminus T_{\mz}$ is the set of active items with respect to the aggregate consumption $\mz=\sum_{i\in\agents} \mx_i$. 
 Every agent chooses their most preferred active item $j$. Let $\gamma_j$ denote the number of agents who have $j$ as their favorite. These agents start eating these items at the same rate up to a quantity $\alpha=\max\left\{\alpha>0\, : \, \sum_{i\in\agents} \mx_i+\alpha \gamma\in\fP(\rho)\right\}$. At this point, a new tight item with respect to the increased aggregate consumption must appear. This process is repeated until every agent has consumed a full unit or no active items remain. 

A formal description of the algorithm is given in 
Algorithm~\ref{alg:submodular-eating}; this corresponds to the Extended PS algorithm in \cite{Fujishige2018} for the case when all agents have the same demand $1$.

\begin{algorithm}[htb!]
\caption{Submodular Eating Algorithm }\label{alg:submodular-eating}
\KwIn{A submodular randomized assignment instance with agents $\agents$, items $\items$, preferences $(\succ_i)_{i\in \agents}$, and value oracle for the submodular function $\rho\, :\, 2^\items\to\Rp$.}
\KwOut{Fractional allocation $\mx\in\Rp^{\agents\times\items}$.}

Initialize $\mx \gets 0$  and $\mz\gets\0$ \;
\While{$\sum_j z_j<n$ \textbf{and} $T_{\mz} \subsetneq \items$,}{
    \lForEach{$i \in \agents$}{
       $j^*_i\gets$ agent $i$'s  most preferred active item in $ \items\setminus T_{\mz}$
    }
    \ForEach{$j \in \items$}{
        $\gamma_j \gets|\{i\in\agents ~|~ j^*_i = j\}|$\tcp*[r]{number of agents who selected item $j$}
    }
    
    Compute $\alpha \gets \max\{\alpha > 0 : \mz + \alpha \gamma \in \fP(\rho)\}$\;
    
    \lForEach{$i \in \agents$}{
        $\mx_i \gets \mx_i+\alpha \chi_{j(i)}$
    }
    Update $\mz \gets \sum_{i \in \agents} \mx_i$ and $T_{\mz}$ accordingly\;
}
\Return{$\mx$}
\end{algorithm}

An efficient implementation of the algorithm needs to maintain the set $T_{\mz}$, and compute the increment step $\alpha>0$. These can all be done efficiently in the standard value oracle model for $\rho$, see \cite{Fujishige2018}. In particular, computing $T_{\mz}$ corresponds to computing the unique maximal minimizer of the submodular function $\rho(S)-\sum_{j\in S}z_j$.
For the output, the following property follows by noting that there are no more active items if and only if the set $\items$ is tight with respect to $\sum_{i\in\agents} \mx_i$.

\begin{proposition}
At the end of the algorithm, the total amount of items fractionally assigned is $\sum_{i\in\agents, j\in \items} x_{ij}=\min\{n,\rho(\items)\}$. 
\end{proposition}

\subsection{Random assignment problem}
We first consider the random assignment problem, where each item has a unit supply.
We prove the following, slightly tighter version of Theorem~\ref{theorem:eating-guar}. Whereas Theorem~\ref{theorem:eating-guar}  also follows from the more general submodular result Theorem~\ref{thm:submod-eating-upper} proved in Section~\ref{sec:submodular-upper-proof}, we include the simpler proof for convenience.
\begin{theorem}
    The PS mechanism approximates the maximum Nash Welfare within a factor $H_n$. Consequently, it is $H_n$-approximately Pareto efficient, where $H_n = \sum_{k=1}^n \frac{1}{k}$ denotes the $n$-th harmonic number, satisfying $\ln (n) \le H_n \le \ln (n) + 1$.
\end{theorem}
\begin{proof} 
Let $\mathbf{x} = \{ \mathbf{x}_i \}_{i \in \agents}$ be an arbitrary matching allocation and $\mathbf{y} = \{ \mathbf{y}_i \}_{i \in \agents}$ be the allocation produced by the PS mechanism. Let $\items_{c}$ be the set of items that are completely consumed during the execution of the PS mechanism, and let $e_t$ denote the $t$-th good to be fully consumed.  

In the PS mechanism, agents consume their most preferred available goods at a constant rate. Since at most $n$ agents can consume goods simultaneously, the $t$-th good $e_t$ cannot be fully exhausted before time $t / n$. It follows that for every agent $i$:
$$u_i(\mathbf{y}_i) \geq  \begin{cases}
    u_{i e_t} \cdot t / n & \text{for all }e_t \in \items_{c} \\
    u_{i j}  & \text{for all } j \notin \items_{c}
\end{cases}.$$

Rearranging these inequalities to bound the utility contribution of each good $j$ in the arbitrary allocation $\mathbf{x}$, we obtain
\begin{align*}
    u_i(\mathbf{x}_i) = \sum_j u_{i j} x_{i j} \leq u_i(\mathbf{y}_i) \left[ \sum_{e_t \in \items_{c}} \frac{n}{t} x_{i e_t} + \sum_{j \notin \items_{c}} x_{i j} \right].
\end{align*}

Dividing by $u_i(\mathbf{y}_i)$ and summing over all agents yields
\begin{align*}
    \sum_{i \in \agents} \frac{u_i(\mx_i)}{u_i(\my_i)} &\leq \sum_{i \in \agents}  \left[ n \sum_{e_t \in  \items_{c}}\frac{x_{i e_t}}{t} + \sum_{j \notin \items_{c}} x_{ij}\right]\\ & \leq n \left(1 + \frac{1}{2} + \cdots + \frac{1}{|\items_{c}|} \right) + (n - |\items_{c}|)\\ & \leq n \cdot H_n\ .
\end{align*}
The second inequality holds as $\sum_i x_{i e_t} \leq 1$ for all $e_t \in \items_{c}$ and $\sum_{i \in \agents} \sum_{j \notin \items_{c}} x_{ij} \leq n - |\items_{c}|$.

To bound the Nash Welfare, we apply AM-GM inequality:
\begin{align*}
    \left[\prod_{i \in \agents} \frac{u_i(\mx_i)}{u_i(\my_i)}\right]^{1 / n} \leq  \frac{1}{n} \sum_{i \in \agents} \frac{u_i(\mx_i)}{u_i(\my_i)} \leq H_n.
\end{align*}
The result follows.
\end{proof}

\subsection{Submodular randomized assignment problem}\label{sec:submodular-upper-proof}
In this section, we prove Theorem~\ref{thm:submod-eating-upper}, which also includes Theorem~\ref{theorem:eating-guar} as a special case.
 We will use the next simple lemma that follows by noting that sum on the left hand side gives a lower bound to the integral
$\int_{p_1}^{p_t} \frac{1}{x}dx$. 
\begin{lemma}\label{lem:ratio-bound} Let $0<p_1<p_2<\ldots<p_k$. Then,
$\sum_{t=2}^{k} (p_t-p_{t-1})/{p_t}\le \ln\left({p_k}/{p_1}\right)$.
\end{lemma}

\begin{proof}[Proof of Theorem~\ref{thm:submod-eating-upper}]
Let $\fP=\fP(\rho)$ be the submodular polyhedron.
Let $\mx\in \Rp^{\agents\times \items}$ be any feasible allocation and $\my\in \Rp^{\agents\times \items}$ the allocation generated by the eating algorithm. 
Let $T_1\subseteq T_2\subseteq \ldots\subseteq T_k\subseteq \items$ denote the sequence of tight sets at the end of each iteration of  Algorithm~\ref{alg:submodular-eating}. Thus, $y(T_t)=\rho(T_t)$ for each $t\in[k]$. Let $S_t\coloneqq T_t\setminus T_{t-1}$ be the items that are active in iteration $t-1$ but not in iteration $t$. 
For every item $j\in\items$, let us define
\[
\theta_j\coloneqq\begin{cases}
\rho(T_t)&\mbox{if }j\in S_t\, ,\\
\min\{n, \rho(\mathcal{E})\}&\mbox{if } j\in \items\setminus T_k\, .
\end{cases}
\]
\begin{claim}
    For all agents $i\in\agents$ and items $j\in J$,
    $u_i(\my_i) \geq u_{ij} \cdot{\theta_j}/n$
\end{claim}
\begin{proof}
Assume 
 $j\in \items\setminus T_k$. Then, item $j$ was available throughout the algorithm. Hence, agent $i$ was always able to choose items of utility at least $u_{ij}$. Since $T_k\subsetneq \items$, all agents consumed a full unit, and therefore $ u_i(\my_i) \geq u_{ij}$ as claimed. Assume next
$j\in S_t$ for some $t$. Then, item $j$ was available until $T_t$ became tight. Up to this point, the total consumption was $\rho(T_t)=\theta_j$. Since everyone is eating at equal speeds, everyone received at least $\theta_j/n$ fractional units in total. For agent $i$, all these items had value at least $u_{ij}$, proving the claim.
\end{proof}
By the above claim, ${n}/{\theta_j}\ge {u_{ij}}/{u_i(\my_i)}$ for all $i\in\agents$ and $j\in\items$. From this, we get
\begin{align*}
     n  \sum_{j \in\items}\frac{x_{ij}}{\theta_j}\geq \frac{\sum_{j\in \items} u_{ij} x_{ij} }{u_i(\my_i)}=\frac{u_i(\mx_i)}{u_i(\my_i)}\, ,\quad\forall i\in \agents\, .
\end{align*}
Let $\mz\coloneqq \sum_{i\in\agents}\mx_i\in \fP$
denote the total consumption in the allocation $\mx$. Summing up over all agents, 
\[
    \frac{1}{n}\sum_{i\in\agents} \frac{u_i(\mx_i)}{u_i(\my_i)}\le \sum_{j\in \items} \frac{z_j}{\theta_j}\, .
\]
Let us denote
\[
\gamma_t\coloneqq\begin{cases}
\sum_{j\in S_t} z_{j}&\mbox{if } t\in [k]\, ,\\
\sum_{j\in \items\setminus T_k} z_{j}&\mbox{if } t=k+1\, .
\end{cases}
\]
With this notation,
\begin{equation}\label{eq:v-t-upper}
   \frac{1}{n}\sum_{i\in\agents} \frac{u_i(\mx_i)}{u_i(\my_i)} \le \frac{\gamma_{k+1}}{\min\{n, \rho(\mathcal{E})\}}+  \sum_{t=1}^{k} \frac{\gamma_t}{\rho(T_t)}\, .
\end{equation}
The feasibility $\mz \in \fP(\rho)$ implies that for $t\in[k]$,  $\sum_{\ell=1}^t \gamma_\ell=z(T_t)
 \le \rho(T_t)$.
Subject to these constraints, the RHS in \eqref{eq:v-t-upper} is maximized when $\gamma_t=\rho(T_t)-\rho(T_{t-1})$ for $t>1$ and $\gamma_1=\rho(T_1)$.
Further, note that $\gamma_{k+1}\le \min\{n, \rho(\mathcal{E})\} - \rho(T_k)$. Therefore,
\begin{equation}\label{eq:v-t-upper-2}
 \frac{1}{n}\sum_{i\in\agents} \frac{u_i(\mx_i)}{u_i(\my_i)} \le \frac{\rho(T_1)}{\rho(T_1)}+ \sum_{t=2}^{k} \frac{\rho(T_t)-\rho(T_{t-1})}{\rho(T_t)} + \frac{\min\{n, \rho(\mathcal{E})\} - \rho(T_k)}{\min\{n, \rho(\mathcal{E})\}}\, .
\end{equation}
We invoke Lemma~\ref{lem:ratio-bound} to get
\begin{equation}\label{eq:v-t-upper-3}
      \frac{1}{n}\sum_{i\in\agents} \frac{u_i(\mx_i)}{u_i(\my_i)}\le 1+\ln \frac{\min\{n, \rho(\mathcal{E})\}}{\rho(T_1)}\le \ln(\min\{n,\rho(\mathcal{E})\})+1\, .
\end{equation}
The final inequality follows because due to our assumption that $\rho(\{j\})\ge 1$ for all $j\in \items$. 
Let $\beta\coloneqq \ln(\min\{n,\rho(\mathcal{E})\})+1$ be the right hand side.
Finally, by the AM-GM inequality,
\begin{align*}
    \left[\prod_{i\in\agents} \frac{u_i(\mx_i)}{u_i(\my_i)}\right]^{\frac{1}{n}} \leq \frac{1}{n}\cdot \sum_{i\in\agents} \frac{u_i(\mx_i)}{u_i(\my_i)} \leq  \beta\, .
\end{align*}
This shows that $\my$ attains $1/\beta$-fraction  of the maximum Nash welfare. It then follows immediately that $\my$ is $\beta$-Pareto efficient.
\end{proof}


\section{Finding approximately envy-free and approximately Pareto-efficient allocations}\label{sec:convex_program}
\subsection{Model and preliminaries}
\paragraph{Fisher markets and competitive equilibrium} In the Fisher market model, we are given a set of $n$ agents $\agents$  and a set of $m$ divisible items $\items$, with $\supply_j\in \Rp$ units available of each $j\in \items$. Every agent $i\in\agents$ has a monotone non-decreasing convex utility function $u_i\,:\,\Rp^\items\to\Rp$ with $u_i(\0)=0$, and a budget $b_i\ge 0$.
A fractional allocation is a vector $\mx\in\Rp^{\agents\times \items}$, and a   \emph{price vector} is a vector $\mp\in\Rp^\items$. 
The \emph{demand correspondence} of agent $i$ is 
 \begin{equation}\label{def:demand}\demand{}_i(\mathbf{p}) \defeq \arg\max_{\mx_i \in\Rp^\items}\left\{ u_i(\mx_i) \, :\, \pr{\mp}{\mx_i}\le b_i\right\} \, .
  \end{equation}

\begin{Definition}[Competitive equilibrium]\label{eq:def-comp} Given a Fisher market instance as above, the allocations and prices $(\mx, \mp)$ form a \emph{competitive (market) equilibrium} if the following conditions hold: 
\begin{enumerate}[(i)]
\setlength\itemsep{0em}
    \item every agent gets an optimal utility at these prices: $\mx_i\in \demand{i}(\mathbf{p})$ for every agent $i\in\agents$. 
    \item no item is oversold: $\sum_i x_{ij} \leq \supply_j$ for all $j\in\items$.
    \item every item with positive price is fully sold:  $\sum_i x_{ij} = s_j$ if $p_j>0$.
\end{enumerate}
\end{Definition}

\paragraph{Matching market equilibrium}
We next describe the Hylland--Zeckhauser matching market model [\citeyear{hylland1979efficient}]. This is similar, but notably different from the Fisher market model above. Again, we are given a set of $n$ agents $\agents$  and a set of $m$ divisible items $\items$, with $\supply_j\in \Rp$ units available of each $j\in \items$. The agents have linear utility functions $u_i(\mx_i)=\sum_{j\in\items} u_{ij}x_{ij}$ as above, but each of them is allowed to buy at most one fractional unit in total. Thus, for a given price vector $\mp\in\Rp^\items$, the \emph{demand correspondence} of agent $i$ is
 \begin{equation}\label{def:demand-matching}\demand{}_i^{\mathrm{HZ}}(\mathbf{p}) \defeq \arg\max_{\mx_i\in\Rp^\items}\left\{ \sum_j u_{ij} x_{ij} \, :\, \sum_{j\in\items} x_{ij} \leq 1 \text{ and } \pr{\mp}{\mx_i}\le 1\right\} \, .
  \end{equation}

\begin{Definition}[Matching market  equilibrium] Given a matching market instance as above, the allocations and prices $(\mx, \mp)$ form a \emph{matching market equilibrium} if {\em (i), (ii)}, and {\em (iii)} in Definition~\ref{eq:def-comp} hold, with the demand system $\demand{}_i^{\mathrm{HZ}}(\mp)$ as in \eqref{def:demand-matching}.
\end{Definition}

We note that the original model \cite{hylland1979efficient} requires that each agent receives exactly $1$ unit, whereas in \eqref{def:demand-matching} we have the weaker requirement of at most $1$ unit. The next proposition shows that when the supply vector is sufficient to satisfy all agents, these two notions coincide. 
\begin{proposition}
For any linear utilities $(u_{ij})_{i\in\agents, j\in\items}$ and any supplies $\supply_j$, a matching market equilibrium exists. Furthermore, if $\sum_j \supply_j \geq n$, then there exists a market equilibrium in which every agent receives a full allocation, i.e., $\sum_j x_{ij} = 1$ for all agents $i$.
\end{proposition}

\begin{proof}
    The existence of a market equilibrium follows from \cite{ garg2022approximating,hylland1979efficient}. For the second part, assume $\sum_j \supply_j \geq n$, and consider a matching market equilibrium $(\mx, \mp)$. We may assume that if $\sum_{j\in\items} x_{ij} < 1$ for any agent $i$, then $\sum_{i\in\agents} x_{ij}=\supply_j$ for all items with $p_j=0$; otherwise, we could increase the allocation $x_{ij}$ between such pairs of agents and items, while maintaining the equilibrium conditions.
    For a contradiction, assume $\sum_{j\in\items} x_{ij} < 1$ for some agent $i$. By the assumption $\sum_j \supply_j \geq n$, there is an item $j$ such that $\sum_{i\in\agents} x_{ij}<\supply_j$. By property \emph{(iii)}, we must have $p_j=0$, but this gives a contradiction. 
\end{proof}

\subsection{Envy-free and approximately Pareto-efficient allocations via Nash welfare}\label{sec:envy-nash}
The main goal of this section is to prove Theorem~\ref{thm:envy-main} for the more general submodular setting. Recall the definitions of polymatroids~\eqref{eq:submodu-polytope} and feasible assignment vectors~\eqref{eq:submod-allocation} for the submodular assignment problem from Section~\ref{sec:sub-prel}. 

Consider the following convex program that maximizes Nash welfare over envy-free feasible allocations. Here, we can enforce envy-freeness by adding a linear inequality for every pair of agents.

\begin{equation}
\begin{aligned}
    \max_{\mx \in \Rp^{\agents \times \items}} \quad & \sum_{i\in\agents} \log \sum_{j\in\items} u_{ij}x_{ij} \\
    \text{s.t.} \quad 
    & \sum_{j\in\items} u_{ij}x_{ij} \geq \sum_{j\in\items} u_{ij}x_{kj}\quad \forall i, k \in \agents\, , \\
    & \sum_{j \in \items} x_{ij} \leq 1 \quad \forall i \in \agents\, , \\
    & \sum_{i \in \agents} \mx_i \in \fP(\rho).  \label{cp::EF-approximatePO-1}
\end{aligned}
\end{equation}

We show the following theorem, which is a more formal restatement of Theorem~\ref{thm:envy-main}.
\begin{theorem}\label{thm:envy-main2}
Let $\rho\,:\,2^{\items}\to \Rp$ be a submodular function and  $u\in\Rp^{\agents\times\items}$ denote the utility values. Then, any optimal solution $\mx\in\R^{\agents\times\items}$  to the convex program \eqref{cp::EF-approximatePO-1} is envy-free and $\ee^{1 / \ee}$-approximates the maximum Nash welfare. Consequently, $\mx$ is  $\ee^{1 / \ee}$-Pareto efficient. Moreover, for any $\varepsilon>0$, an envy-free and $(\ee^{1 / \ee}-\varepsilon)$-Pareto efficient allocation can be computed by an algorithm where the number of arithmetic operations and value oracle queries to the function $\rho$ is polynomial in the encoding length of  $u\in\mathbb{Q}_{\geq0}^{\agents\times\items}$ and $\log(1/\varepsilon)$. 
 \end{theorem}

As a first step, we show a `price of anarchy' type result in the following theorem, showing that the matching market equilibrium approximates the maximum Nash welfare by a constant factor. A closely related statement was recently shown by \citet[Theorem 4.1]{garg2025approximating} in a more general setting of Fisher market with concave utilities; however that model imposes no constraints on the agents' feasible allocations. We adapt their proof to our setting and present it in Appendix~\ref{sec:poa-appendix}.
\begin{restatable}{theorem}{poathm}\label{thm:poa}
    Given a matching market equilibrium instance with available amount $\supply_j$ of item $j$, let $(\mx,\mp)$ be a matching market equilibrium. Then, $\NSW(\mx)\ge \left(\frac{1}{\ee}\right)^{1 / \ee} \NSW(\my)$
    for any feasible allocation $\my\in\Rp^{\agents\times\items}$.
\end{restatable}

\begin{proof}[Proof of Theorem~\ref{thm:envy-main2}]
Let $\mx$ be the optimal solution to \eqref{cp::EF-approximatePO-1}. This (along with all other feasible solutions) is feasible and envy-free by the first set of constraints. Let us now show that it approximates the maximum Nash welfare by a factor $ \ee^{1 / \ee}$. Let $\my$ be the maximum Nash welfare allocation, that is, the optimal solution to

\begin{equation}\label{cp::max-NSW-submod}
 \max_{\my \in \Rp^{\agents \times \items}} \left\{ \sum_{i\in\agents} \log \sum_{j\in\items} u_{ij}x_{ij} \,\,\middle| 
    \,\,\sum_{j \in \items} y_{ij} \leq 1 \quad \forall i \in \agents\, , \quad
     \sum_{i \in \agents} \my_i \in \fP(\rho)\right\}
\end{equation}

Let $\supply=\sum_{i \in \agents} \my_{i}$. Thus, $\supply\in \fP(\rho)$. We consider a matching market with supplies $\supply_j$. Let $(\mz,\mp)$ be a matching market equilibrium. According to Theorem~\ref{thm:poa}, $\NSW(\mz)\ge (1 / \ee)^{1 / \ee}\NSW(\my)$. Now, $\mz$ is a feasible solution to \eqref{cp::EF-approximatePO-1}, because $\sum_{i\in \agents} \mz_i\le \supply\in \fP(\rho)$, and $\mz$ is envy-free because it is a competitive equilibrium. Hence, 
\[\NSW(\mx)\ge \NSW(\mz)\ge \left(\frac{1}\ee\right)^{1 / \ee}\NSW(\my).\]
The $(1 / \ee)^{1 / \ee}$-Pareto efficiency of $\mx$ is then immediate. 

Finally, for polynomial-time computability, we can use the ellipsoid method to solve \eqref{cp::EF-approximatePO-1} up to accuracy $\varepsilon$ for any $\varepsilon>0$. For the separation oracle, we need to be able to separate $\sum_{i \in \agents} \mx_i \in \fP(\rho)$; this can be done using submodular function minimization. We refer the reader to \cite{gls} for the details.
\end{proof}

\subsection{Approximately envy-free and Pareto-efficient allocations in general Fisher markets}\label{sec:Fisher}
Analogous to the random assignment problem, a natural question is whether it is possible to efficiently find envy-free but approximately Pareto-efficient allocations in the more general Fisher market setting. In this section, we answer this question by proving Theorem~\ref{thm:Fisher}, showing how to compute in polynomial-time an allocation that is $(1+\varepsilon)$-approximately envy-free and $(2+2\varepsilon)$-approximately maximum Nash welfare.

We consider a Fisher market with a set $\agents$ of $n$ agents, a set $\items$ of $m$ items, and monotone non-decreasing concave utilities $u_i\,:\,\Rp^\items\to\Rp$. We assume that the following \emph{oracle} is available:

\begin{center}
\fbox{
\begin{minipage}{0.8\textwidth}
\noindent
{\sf Oracle}  \\
\textbf{Input:} Agent $i$, a real number $s$, and two bundles $\mx$ and $\my$. \\
\textbf{Output:} The smallest $\alpha \geq 0$ such that $u_i(\mx + \alpha \my) \geq s$.
\end{minipage}
}
\end{center}

Given an evaluation oracle to $u_i$, we can implement this using binary search on the values of $\alpha$ to arbitrary $\varepsilon$-accuracy. For simplicity of the description, we assume an exact oracle is available.

In the remainder of this section, we prove Theorem~\ref{thm:Fisher}, which we restate more formally below:

\begin{theorem}\label{thm:Fisher2}
Given a Fisher market instance with $n$ agents $\agents$, $m$ items $\items$, monotone concave utility functions $u_i\,:\, \R^\items\to\Rp$, unit budgets, and a parameter $\varepsilon>0$, there exists an algorithm that computes a $(1+\varepsilon)$-envy-free and $(2+2\varepsilon)$-approximate maximum Nash welfare allocation. The number of arithmetic operations and value oracle calls to the utility functions made by the algorithm is polynomial in $n$, $m$, and $1/\varepsilon$. Such an allocation is also $(2+2\varepsilon)$-approximately Pareto efficient.
\end{theorem}

A key challenge is that the approach used for the random assignment problem in Theorem~\ref{thm:envy-main} of solving the convex program~\eqref{cp::EF-approximatePO-1} does not work: the constraint $u_i(\mx_i)\ge u_i(\mx_k)$ is convex only for linear utility functions. 

We use the following approach: first, the algorithm constructs a partial allocation with the desired properties. Starting from a Nash welfare maximizing allocation $\mx^\star$, each agent $i$ is initially assigned the fractional bundle $\mx_{k}^\star/n$ that maximizes $i$'s utility. We then let agents repeatedly swap their bundles for unallocated fractional bundles until there is no $(1+\varepsilon)$-envy between agents' bundles and the unallocated bundles. In a second stage, we gradually extend the partial allocation to a complete allocation of the items while maintaining approximate envy-freeness.

This approach is inspired by the work of \citet{barman2024}, who showed the existence of an EF1 and $2$-approximate maximum Nash welfare allocation of \emph{indivisible items}. That is, we do not allow randomization, but need to fully allocate each item to the same agent; EF1 is a relaxation of envy-freeness for the indivisible setting. Our algorithm can be seen as a continuous analogue of their algorithm. Due to the divisibility of items, our algorithm is considerably simpler.

Given an allocation $\mx\in\R^{\agents\times\items}$, the \emph{envy graph} $G = (V, E)$ is constructed as follows: each node represents an agent, and there is a directed edge $(i, j) \in E$ if and only if agent $i$ envies agent $j$, that is, $u_i(\mx_j) > u_i(\mx_i)$.

\paragraph{Cycle Elimination} 
An important subroutine eliminates cycles from the envy graph.
Let $i_0 \rightarrow i_1 \rightarrow \cdots \rightarrow i_k \rightarrow i_0$ be a cycle in the envy graph. We can  eliminate it by rotating the allocations: assign $\mx_{i_1}$ to $i_0$, $\mx_{i_2}$ to $i_1$, $\ldots$, and $\mx_{i_0}$ to $i_k$. Such a reallocation reduces the number of edges in the envy graph. Finding an envy cycle and a reallocation can be implemented in polynomial time.

\subsection{Stage I: constructing a fractional allocation}
The first stage of our algorithm is a partial allocation that is  $(1+\varepsilon)$-envy-free and $2(1+\varepsilon)$-approximates the maximum Nash welfare. We begin with a Nash welfare maximizing allocation $\mx^\star$.\footnote{We can obtain by convex programming an $(1+\varepsilon)$-approximate maximum Nash welfare allocation; the running time polynomially depends on $\log(1/\varepsilon)$. For simplicity, we describe the algorithm with an exact maximum Nash welfare allocation; one can adjust it to use an approximate solver by choosing a smaller $\varepsilon$-value.
} 
We consider the $n$ bundles $\mx^\star_i$ in this allocation, and initialize the process by letting each agent get their favorite $1/n$-portion from one of these bundles. The remaining portions of each bundle are denoted by $\{\my_i\}_{i\in\agents}$.

During the process, we keep updating each remaining bundle $\my_i$ if any agent prefers $\my_i$ over their current allocation. Specifically, we divide $\my_i$ into the largest possible portion such that no agent would envy this portion more than $(1 + \varepsilon)$-times their current allocation. This portion is then assigned to the agent who envies it the most. In exchange, we reclaim that agent's previous allocation and update each of the remaining bundles $\{\my_i\}_{i\in\agents}$ accordingly. We proceed until there are no more remaining bundles $\my_i$ preferred by some agent to their current allocation.

\begin{algorithm}[ht]
\caption{Partial allocation}\label{alg:partial}
\KwIn{Nash welfare maximizing allocation $\mx^\star\in\Rp^{\agents\times \items}$ and parameter $\varepsilon > 0$}
\KwOut{A $(1 + \varepsilon)$-envy-free and $2(1 + \varepsilon)$-approximate maximum Nash welfare partial allocation}

Initialize $\{\my_i\}_{i \in \agents} \gets \{\mx^\star_i\}_{i \in \agents}$ \;
 \lForEach{$i \in \agents$}{
$h(i) \gets \arg \max_k u_i\left(\frac{1}{n} \my_k\right)$,  $\mathbf{x}_i \gets \frac{1}{n} \my_{h(i)} $}
 \ForEach{$k \in \agents$}{%
 $$\my_k \gets \left(1 - \frac{\left|\{i ~|~ h(i) = k \}\right|}{n}\right) \my_k$$%
 }
\While{ $\exists i,\ell\in \agents$ such that $u_i(\my_{\ell}) >   (1 + \varepsilon)u_i(\mx_i)$}{
\lForEach{$i\in\agents$}{
    compute the minimum $\alpha_i$ such that
    $u_i(\alpha_i \cdot \my_{\ell}) \geq  (1 + \varepsilon)u_i(\mathbf{x}_i)$}
    $k \gets \arg \min_{i\in\agents} \{\alpha_i\}$ and $\tilde{\mx}_{k} = \alpha_k \cdot \my_{\ell}$\;
    $\my_{h(k)} \gets  \my_{h(k)} + \mx_{k}$ \tcp*[r]{return $k$'s previous allocation}
     $\mathbf{x}_{k} \gets \tilde{\mx}_{k}$; \, $\my_{\ell} \gets  \my_{\ell} - \tilde{\mx}_k$ and $h(k) \gets \ell$ \;
}
\Return{Partial allocation $\mx$}
\end{algorithm}

\begin{lemma}
    Algorithm~\ref{alg:partial} terminates after at most $\frac{n^2}{\varepsilon}$ iterations.
\end{lemma}
\begin{proof}
    Consider the potential function $\sum_i \frac{u_i(\mx_i)}{u_i(\items)}$. Initially, this potential is at least $\frac{1}{n}$ as $u_i(\mx_i) \geq \frac{1}{n^2} u_i(\items)$ by concavity, and the potential is always bounded above by $n$. In each iteration, the potential increases by at least $\varepsilon \cdot \frac{u_i(\mx_i)}{u_i(\items)} \geq \varepsilon \cdot \frac{1}{n^2}$. Therefore, the number of iterations is at most $\frac{n - \frac{1}{n}}{\varepsilon \cdot \frac{1}{n^2}} \leq \frac{n^3}{\varepsilon}$.
    Hence, the algorithm terminates after at most $\frac{n^3}{\varepsilon}$ iterations.
\end{proof}

\begin{lemma}
    Algorithm~\ref{alg:partial} outputs  $(1+\varepsilon)$-envy-free partial allocation $\mx$ that   $2(1+\varepsilon)$-approximately maximizes Nash welfare.
\end{lemma}

\begin{proof}
It is easy to check that $(1 + \varepsilon)$-envy freeness is maintained throughout the algorithm.
To show the Nash welfare guarantee, note that every agent gets a proportional fraction of one of the original bundles $\mx^\star_k$.
Let $t_k$ denote the number of parts bundle $\mx^\star_k$ is divided into. This includes the number of agents that receive a part of bundle $k$, i.e., $h(i)=k$, and plus one if $\my_k\neq\emptyset$ at the end of the algorithm. Thus,  $\sum_{k\in\agents} t_k \leq 2n$. At the end of the algorithm, the Nash welfare of the optimal allocation $\mx^\star$ satisfies:
\begin{align*}
    \prod_{k\in\agents} \left( u_k(\mx^\star_k) \right)^{\frac{1}{n}} 
    &\leq \prod_{k\in\agents} \left( u_k(\my_k) + \sum_{i: h(i) = k} u_k(\mx_i) \right)^{\frac{1}{n}} \\
    &\leq \prod_{k\in\agents} \left( t_k (1 + \varepsilon) u_k(\mx_k) \right)^{\frac{1}{n}} \\
    &\leq 2(1 + \varepsilon) \prod_{k\in\agents} \left( u_k(\mx_k) \right)^{\frac{1}{n}}\, .
\end{align*}
Here, the first inequality used subadditivity of $u_k$ that follows by concavity and $u_i(\0)=0$. The second inequality uses that by approximate envy-freeness, $u_k(\my_k)\le (1+\varepsilon) u_k(\mx_k)$ and $u_k(\mx_i)\le (1+\varepsilon) u_k(\mx_k)$ for all $i\in\agents$.
The final inequality uses the AM-GM inequality and that $\sum_{k}t_k\le 2n$.\end{proof}

\subsection{Stage II: completing the allocation}
In the second algorithm, we begin with a partial allocation that $2(1+\varepsilon)$-approximates the maximum Nash welfare and is  $(1+\varepsilon)$-envy-free, and extend it to a complete allocation preserving these properties. 

At each iteration, we construct the envy graph based on the current allocation and eliminate any cycles by rotating the allocations among the agents involved. We choose one agent who is not currently envied by anyone and allocate the remaining items until this agent becomes envied by some other agent by a factor greater than $(1 + \varepsilon)$.

\begin{algorithm}[H]
\caption{Extending Partial  to Complete Allocation}
\KwIn{A partial allocation $\mx\in\Rp^{\agents\times \items}$ that is $(1 + \varepsilon)$-envy-free and $2(1 + \varepsilon)$-approximate maximum Nash welfare partial allocation}
\KwOut{A complete allocation $\mz\in\Rp^{\agents\times \items}$  of all items with the same properties}
$\mz\gets\mx$
\While{there exist unallocated items}{
    Construct the envy graph based on the current allocation $\mz$ \;
    \While{a cycle exists in the envy graph}{
        Eliminate the cycle by rotating the allocations among the agents in the cycle \;
    }
    Select an agent $i\in \agents$ who is not envied by any other agent \;
    Allocate additional goods to agent $i$ until some agent $k$ envies $i$ by  a factor of $(1 + \varepsilon)$ \;
}
\Return{$\mz$}
\end{algorithm}
\begin{lemma}
    The algorithm terminates after at most $\frac{n^3}{\varepsilon}$ iterations.
\end{lemma}
\begin{proof}
    Consider the potential function $\sum_{i,k\in\agents} \frac{u_i(\mx_k)}{u_i(\items)}$. Since the utility function is non-decreasing, each iteration increases the potential by at least $\varepsilon\cdot \frac{u_k(\mx_i)}{u_k(\items)} \geq \varepsilon/{n^3}$. The potential function is upper bounded by $n$. Therefore, this gives an upper bound $\frac{n^4}{\varepsilon}$ on the total number of iterations.
\end{proof}
Theorem~\ref{thm:Fisher2} easily follows by the above statements.
\begin{remark}
The approximation ratio, $(2 + 2 \varepsilon)$, guaranteed in Theorem~\ref{thm:Fisher2} for Fisher markets is larger than the approximation ratio, $(1 / \ee)^{1 / \ee}$, guaranteed in Theorem~\ref{thm:envy-main} for matching markets. Moreover, the algorithm in Theorem~\ref{thm:Fisher2} is not applicable to matching markets. Even though the first stage can similarly find an $(1+\varepsilon)$-envy free and $(2+2\varepsilon)$-approximate maximum Nash welfare allocation, the second stage is not directly applicable. This is because the set of possible allocations of the agents is constrained by $\sum_{j\in \items} x_{ij}\le 1$; so, we cannot allocate additional items to agents who already have one fractional unit. 
\end{remark}

\begin{remark}
As mentioned in the Introduction, \cite{garg2025approximating} showed that maximizing Nash welfare yields a Pareto efficient and 2-approximate envy-free allocation. This raises the natural question: can we obtain an envy-free and approximately Pareto efficient allocation starting from it? One possible approach is to first reduce each agent's bundle to achieve exact envy-freeness, and then redistribute the remaining goods. Ideally, during this procedure, each agent would lose at most a constant fraction of their utility. However, this guarantee cannot always be ensured, as the following example demonstrates.
\begin{example}
Consider a case with two agents and a single item. Both agents have unit budgets. The utility function of the first agent is linear:
\[
u_1(t) = t.
\]
The utility function of the second agent is defined as:
\[
u_2(t) = 
\begin{cases}
\alpha t& \text{for } t \leq \frac{1}{\alpha}, \\
1 + \beta \left(t - \frac{1}{\alpha}\right) & \text{for } t > \frac{1}{\alpha},
\end{cases}
\]
where $\alpha$ is a large constant and $\beta \leq \alpha$.

When $\beta = \frac{1 - \frac{1}{\alpha}}{1 - \frac{2}{\alpha}}$, the Nash Welfare maximizing allocation is:
\[
x_1 = 1 - \frac{1}{\alpha}, \quad x_2 = \frac{1}{\alpha}.
\]    
Since utility is increasing, to ensure envy-freeness, the first agent's allocation would need to be reduced to $\frac{1}{\alpha}$. In the limit as $\alpha \rightarrow \infty$, it becomes impossible to guarantee that the first agent retains a constant fraction of the utility.
\end{example}
\end{remark}

\section{Efficiency of the PS mechanism with chores}
\subsection{Model and preliminaries}
We consider the random assignment problem in the chores setting, where $\agents$ denotes a set of $n$ agents and $\items$ a set of $m$ items (chores). In the ordinal model, each agent provides a strict preference orderings $(\succ_i)_{i\in\agents}$ over the chores, without specifying the intensity of their preferences. In the cardinal model, each agent $i\in\agents$ specifies a non-negative disutility $d_{ij}\ge 0$ for each chore $j\in\items$. Given a fractional allocation $\mx\in\Rp^{\agents\times \items}$, the expected disutility of agent $i$ is defined as $d_i(\mx_i):= \sum_{j\in\items} d_{ij} x_{ij}$. The cardinal preferences $(d_{ij})_{j\in\items}$ are \emph{consistent} with the ordinal ordering $\succ_i$, if  $j\succ_i j'$ implies $d_{ij}\le d_{ij'}$. 

A mechanism is \emph{envy-free} in the cardinal model if $d_i(\mx_i)\le d_i(\mx_{i'})$ for any $i,i'\in \agents$. In the ordinal model, envy-freeness is defined via stochastic dominance as in the goods setting: an allocation $\mx$ is envy-free if, for any $i, i'\in \agents$, $\mx_{i'}$ does not stochastically dominates $\mx_i$ with respect to $\succ_i$. 

As in the goods setting, the simultaneous eating algorithm (i.e., the PS mechanism) extends naturally to the chores setting. All agents simultaneously consume their most preferred available chore (i.e., lowest disutility chore) at a uniform rate. When a chore is fully consumed, agents switch to their next most preferred available chore. The process continues until each agent has consumed one unit in total or all chores are exhausted. As in the goods setting, the PS mechanism outputs an envy-free random assignment and is \emph{ordinally-efficient} in the chores setting, i.e., no other randomized assignment stochastically dominates it for all agents. 

In the cardinal model, efficiency is captured by Pareto dominance: an allocation $\mx'$ Pareto dominates $\mx$ if $\sum_j d_{ij}x'_{ij} \le \sum_j d_{ij}x_{ij}, \forall i\in\agents$, with strict inequality for at least one agent. A mechanism is Pareto-efficient if no such domination occurs. The PS mechanism is not Pareto-efficient in the cardinal model in the chores setting either, as seen by a straightforward modification of Example~\ref{example1}. This motivates the notion of \emph{$\gamma$-approximately Pareto efficient lotteries} for chores: For $\gamma\ge 1$, a randomized assignment $\mx\in\R^{\agents\times\items}$ is said to be $\gamma$-approximately Pareto efficient (or $\gamma$-Pareto efficient for short), if there does not exist another assignment $\my$ such that $\gamma d_i(\my_i)\le d_i(\mx_i)$ for all $i\in\agents$, with strict inequality for at least one agent. 

\subsection{Random assignment problem with chores}
In this section, we prove Theorem~\ref{thm:chores-main}.
We assume that $m\ge n$, as the case $m<n$ can be reduced to $m=n$ (see Note~\ref{footnote2} in Section~\ref{sec:intro-chores}).
We start by showing the efficiency guarantee. 
\begin{lemma}
If all disutilities are positive, then the PS is $n$-approximately Pareto-efficient for matching markets with chores.
\end{lemma}
\begin{proof}
    Let $\mx\in\R^{\agents\times\items}$ be any matching allocation and $\my\in\R^{\agents\times\items}$ be the allocation generated by the eating algorithm.
    Suppose that for each agent $i\in\agents$, the eating algorithm uses the following preference ordering: 
    $0 < d_{i,i_1} \leq d_{i,i_2} \leq \cdots \leq d_{i,i_m}$. For each agent $i$, define $t_i$ to be the smallest index such that $\sum_{s = t_i + 1}^{m} x_{i, i_s} < \frac{1}{n}$ and define the set $T_i \defeq \{i_1, i_2, \cdots, i_{t_i}\}$. 
    According to this definition, for any $i$,
    $$d_i(\mx_i)=\sum_{s = 1}^m d_{i, i_s} x_{i, i_s} \geq \sum_{s = t_i}^m d_{i, i_s} x_{i, i_s} \geq  \frac{1}{n} d_{i, i_{t_i}},$$
    where the first inequality becomes equality when $\sum_{s = 1}^{t_i - 1} x_{i, i_s} = 0$, and the second inequality becomes equality only when $\sum_{s = t_i}^m x_{i, i_s} = \frac{1}{n}$. Since $\{\mx_i\}_{i \in \agents}$ allocates one fractional unit, these two conditions cannot hold simultaneously. Therefore, the inequality is strict: \begin{align}
        \sum_{s = 1}^m d_{i, i_s} x_{i, i_s}  >  \frac{1}{n} d_{i, i_{t_i}}. \label{ineq:chores-ti}
    \end{align}
    We now show that there exists an agent $k\in\agents$ such that 
    \begin{equation}\label{eq:agent-k}
    d_{k, k_{t_{k}}} \geq d_k(\my_k)\, .
    \end{equation} Combining this with inequality \eqref{ineq:chores-ti} yields the result.

To achieve this, we claim that there exists some agent $k$ whose allocation under the eating algorithm is fully supported in $T_{k}$. First, by the definition of $t_i$, $\sum_{i\in\agents, j \in T_i} x_{ij} > n - 1$. This implies $|\cup_{i} T_i | \geq n $. We now consider two cases. 
\smallskip
    
    \noindent\textbf{Case $1$:} there exists a chore $j \in \cup_{i} T_i$ such that $\sum_i y_{ij} < 1$. In this case, consider the agent $k$ such that $j \in T_{k}$. Since chore $j$ is not fully consumed, agent $k$ must consume only $j$ or more preferred chores.  Therefore, agent $k$ satisfies \eqref{eq:agent-k}.
    \smallskip
    
    \noindent\textbf{Case $2$:} for every chore $j \in \cup_{i\in\agents} T_i$, $\sum_i y_{ij} = 1$. In this case, since $|\cup_{i\in \agents} T_i | \geq n $, we have $|\cup_{i} T_i | = n $ and all consumption is supported on $\cup_{i} T_i$. 
   Every agent $k$ who only consumes chores in $T_k$ satisfies \eqref{eq:agent-k}.
    Suppose for contradiction that for every agent $k$, there is a chore $j\notin T_k$ with $y_{ik}>0$. Then, at the end of the eating process, all agents must be consuming chores outside of their respective $T_i$. Consider any particular agent $\ell$, who is consuming chore $j\notin T_\ell$. Since $j \in \cup_{i\in\agents} T_i$, then there exists an agent $k$ such that $j \in T_{k}$. We however assumed that all agents are consuming chores outside their $T_i$ sets at this point; in particular, $T_k$ must have been fully consumed at this point. This yields a contradiction.
\end{proof}

Our next example shows that assuming positive disutilities above  is necessary, as zero disutilities  may result in unbounded inefficiency of the eating algorithm.
\begin{example}\label{ex:zero-chore}
    Consider a scenario with two agents and two chores. All disutilities are zero except that Agent $1$ has a disutility of $1$ for Chore $2$. Suppose both agents follow the same priority: they first consume Chore $1$, then Chore $2$. Then, since each chore is split equally when chosen simultaneously, Agent $1$ thus incurs a disutility of $0.5$ (half of Chore $2$), while Agent $2$ incurs $0$ disutility. However, a more efficient allocation exists: assign Chore $1$ entirely to Agent $1$ and Chore $2$ entirely to Agent $2$. This allocation results in disutilities of $(0, 0)$. By definition, the simultaneous eating algorithm is unboundedly inefficient in this case.
\end{example}

We close this section with the proof of the second part of Theorem~\ref{thm:chores-main}, an example when the disutility in the eating algorithm is worse by a factor $n/4$ for all agents.
\begin{example}
    Let $0 < \varepsilon_1 < \varepsilon_2 < \varepsilon_3 \cdots < \varepsilon_n  \ll 1$, and suppose there are $n$ agents and $n$ chores, with  $n$ being even. The disutilities are defined as follows: for the first $n / 2$ agents (for all $i \leq n / 2$), the disutility from each chore is:
    \begin{align*}
        d_{ij} =  \left\{ \begin{array}{rcl}
\varepsilon_j & \mbox{for}
& j < n \\ 1 + \varepsilon_n & \mbox{for} & j = n \\
\end{array}\right. \quad \mbox{for any $i \leq n / 2$;}
    \end{align*} 
    for the remaining $n / 2$ agent (for all $i > n / 2$), the disutility is: 
    \begin{align*}
        d_{ij} =  \left\{ \begin{array}{rcl}
\varepsilon_j & \mbox{for}
& j <  n / 2 \\ 1 + \varepsilon_j & \mbox{for} & j \geq n / 2 \\
\end{array}\right. \quad \mbox{for any $i > n / 2$.}
    \end{align*} 
    
    Under the eating algorithm, each agents consumes $\frac{1}{n}$ fraction of all chores. This results in a disutility of $\frac{1 + \sum_{j=1}^n \varepsilon_j}{n}$ for each of the first $n / 2$ agents and $\frac{n / 2 + 1 + \sum_{j=1}^n \varepsilon_j}{n}$ for each of the remaining $n / 2$ agents.

    Now, we consider an alternative allocation $\mx\in\R^{\agents\times \items}$, where the first $n / 2$ agents equally divide chores $j = n / 2, \cdots, n - 1$, and the remaining $n / 2$ agents equally divide chores $j = 1, \cdots, n / 2 - 1$ and $j = n$. In this allocation, the first $n / 2$ agents receive a disutility of $\frac{\sum_{j = n / 2}^{n-1} \varepsilon_j}{n / 2}$ and the remaining $n / 2$ agents receive a  disutility of $\frac{1 + \varepsilon_n + \sum_{j = 1}^{n / 2 - 1} \varepsilon_j}{n / 2}$. 

    When $\varepsilon_j \rightarrow 0$ for all $j$, the disutility under the eating algorithm becomes at least $n / 4$ worse than $\{\mx_i\}_{i \in \agents}$.
\end{example}

\printbibliography
\appendix
\section{Nash welfare guarantee of matching market equilibria}\label{sec:poa-appendix}
In this appendix, we adapt the proof from \cite{garg2025approximating} to matching markets. 
\poathm*
\begin{proof}
Let $\my\in\Rp^{\agents\times\items}$ be any feasible allocation, i.e., $\sum_{i\in\agents} y_{ij}\le \supply_j$ for any item $j\in\items$ and $\sum_{j\in\items} y_{ij}\le 1$ for any $i\in\agents$.

First, $u_i(\mx_i) \geq \min \left\{1, \frac{1}{\pr{\mp}{\my_i}} \right\}u_i(\my_i)$, as $\frac{\my_i}{\max\{1, \pr{\mathbf{p}}{\my_i}\}}$ is a feasible solution to the demand correspondence $\demand{}_i(\mathbf{p})$, and the utilities are linear.
Thus,
    \begin{align*}
        \sum_{i\in\agents}  \log \frac{u_i(\my_i)}{u_i(\mx_i)} \leq  \sum_{i\in\agents}  \log \max \left\{1, \pr{\mp}{\my_i} \right\}.
    \end{align*}
    Let $b'_i \defeq {\pr{\mp}{\my_i}}$, and let $\agents'$ be the set of agents such that $b'_i \geq 1$, and $B' \defeq \sum_{i \in \mathcal{A}'} b'_i$. Let $\alpha=|\agents'|/n$; by definition, $\alpha\le 1$. From the log sum inequality,
    \begin{align*}
           \sum_{i\in\agents}  \log \frac{u_i(\my_i)}{u_i(\mx_i)} \leq \sum_{i \in \agents'} \log b'_i \leq \alpha n \log \frac{B'}{\alpha n} \, .
    \end{align*}
 Noting also that $B'\le \sum_{i\in\agents} b_i'=\sum_{i\in \agents} \pr{\mp}{\my_i} \le n$, 
    \begin{align*}
         \sum_{i\in\agents} \frac{1}{n} \log \frac{u_i(\my_i)}{u_i(\mx_i)} \leq \alpha \log \frac{1}{\alpha} \leq \frac{1}{\ee}\, .
    \end{align*}
    The statement follows.
\end{proof}

\section{Finding approximately envy-free and Pareto-efficient allocations}\label{appendix:2-ef-po}
In this section, we show that the optimal solution to the following convex program, which maximizes the Nash welfare, is $2$-approximately envy-free and Pareto efficient.
\begin{equation}
\begin{aligned}
    \max_{\my \in \Rp^{\agents \times \items}} \quad & \sum_i \log u_i(\my_i) \\
    \text{s.t.} \quad 
    & \sum_{j \in \items} y_{ij} \leq 1 \quad \forall i \in \agents \\
    & \sum_{i \in \agents} \my_i \in \fP(\rho).  \label{cp::approximateEF-PO}
\end{aligned}
\end{equation}
\begin{lemma}
    The optimal solution to \eqref{cp::approximateEF-PO} yields an allocation that is $2$-approximately envy-free and Pareto efficient.
\end{lemma}
\begin{proof}
    Pareto efficiency follows from the fact that we are maximizing the Nash Welfare. The $2$-envy-freeness of the allocation follows from the results in~\cite{garg2025approximating, troebst2024cardinal}. Since the optimal solution $\my^*$ to \eqref{cp::approximateEF-PO} is also the optimal solution of the following program:
\begin{equation}
    \begin{aligned}
        \max_{\my \in \Rp^{\agents \times \items}} \quad & \sum_i \log  u_i(\my_i) \\
        \text{s.t.} \quad 
        & \sum_{j \in \items} y_{ij} \leq 1 \quad \forall i \in \agents \\
        & \sum_{i \in \agents} y_{ij} \leq \sum_{i \in \agents} y_{ij}^* \quad \forall j \in \items.\label{cp::restrictedSupply}
    \end{aligned}
    \end{equation}
    The framework established in~\cite{garg2025approximating, troebst2024cardinal} proves the optimal solution in \eqref{cp::restrictedSupply} is $2$-envy-free.
\end{proof}
\end{document}